\documentclass{ifacconf}

\usepackage{graphicx}      
\usepackage{natbib}        
\usepackage{amsmath}      
\usepackage{mathrsfs}
\usepackage{comment}

\usepackage{amssymb} 

\newtheorem{theorem}{Theorem}

\usepackage{algorithm}
\newtheorem{assumption}{Assumption}
\newtheorem{remark}{Remark}

\usepackage{booktabs}
\usepackage{algorithmic}

\begin{document}
\begin{frontmatter}

\title{Gradient-Based Adaptive Prediction and Control for Nonlinear Dynamical Systems\thanksref{footnoteinfo}} 
\thanks[footnoteinfo]{This paper was supported by the National Natural Science Foundation of China under Grants T2293770 and 12288201, and the China Postdoctoral Science Foundation 2024M763469.}

\author[1]{Yujing Liu}
\author[1]{Xin Zheng}
\author[1]{Zhixin Liu} 
\author[1]{Lei Guo}

\address[1]{State Key Laboratory of Mathematical Sciences, Academy of Mathematics and Systems Science, 
Chinese Academy of Sciences, Beijing 100190, China 
(e-mail: liuyujing@amss.ac.cn, zhengxin2021@amss.ac.cn, lzx@amss.ac.cn, lguo@amss.ac.cn).}

\begin{abstract}      
This paper investigates gradient-based adaptive prediction and control for nonlinear stochastic dynamical systems under a weak convexity condition on the prediction-based loss. 
This condition accommodates a broad range of nonlinear models in control and machine learning such as saturation functions, sigmoid, ReLU and tanh activation functions, and standard classification models. 
Without requiring any persistent excitation of the data, we establish global convergence of the proposed adaptive predictor and derive explicit rates for its asymptotic performance. 
Furthermore, under a classical nonlinear minimum-phase condition and with a linear growth bound on the nonlinearities, we establish the convergence rate of the resulting closed-loop control error. 
Finally, we demonstrate the effectiveness of the proposed adaptive prediction algorithm on a real-world judicial sentencing dataset. The adaptive control performance will also be evaluated via a numerical simulation.
\end{abstract}


\end{frontmatter}

\section{Introduction}
\label{submission}

Adaptive prediction and control for dynamical systems with structure uncertainties is one of the fundamental problems in control theory and machine learning (cf., \cite{goodwin2014adaptive, chen2012identification, Astrom1995AdaptiveControl}). 
In the era of big data, the effective dimension of data streams can be extremely large and algorithms are often required to operate in an online version under stringent computational constraints.
In this context, stochastic gradient (SG)-type algorithms have become particularly attractive due to their significantly lower computational cost and ability to process data sequentially, compared with many alternative methods (cf., \cite{bottou2018optimization,zhou2020towards}).
Motivated by these progress, it is natural to seek adaptive prediction and control guarantees for SG-based algorithms applied to nonlinear stochastic dynamical systems with unknown parameters.

For general nonlinear dynamical systems, the adaptive prediction problem has been extensively studied under some idealized statistical assumptions or persistent excitation (PE) conditions.
For example, in most of the machine learning and statistics literature, prediction problems are typically formulated within a deterministic optimization or statistical learning framework, where one seeks to minimize an expected loss under the assumption that the data are independent and identically distributed or, more generally, stationary and ergodic (see, e.g., \cite{hazan2007logarithmic,hazan2015beyond,hardt2018gradient}).
Besides these, \cite{lai1991adaptive} established the logarithmic upper bound of the accumulated adaptive prediction error for the nonlinear least squares algorithm under a strong PE data condition, i.e., the matrix $\frac{1}{n}\nabla f(\phi_k,\theta)\nabla f(\phi_k,\theta)^{\tau}$ converges to a positive definite matrix. 
However, such conditions can hardly be verified or satisfied in many practical systems, especially stochastic in systems with feedback control, where the regressor is strongly coupled with past outputs and parameter estimates (cf., \cite{lei2020feedback}).

Fortunately, the convergence of adaptive prediction and control problem has been well established for linear systems or weakly nonlinear systems without any excitation condition on the regressor data.
For example, for linear systems, \cite{guo1995convergence} established the logarithmic upper bound of the least squares (LS)-based accumulated adaptive prediction error for either any bounded random data or for data generated from the closed-loop adaptive control systems under the well-known LS-based self-tuning regulators. 
In addition, for the SG algorithm, the convergence of the prediction error as well as the convergence and optimality analysis of the SG-based adaptive controller have been established (\cite{goodwin1981discrete}).
As for nonlinear systems, under a strong convexity assumption on the mean-square loss, \cite{hazan2007logarithmic} established the logarithmic upper bound of the accumulated adaptive prediction error of the Newton-type algorithm, which can be regarded as an extension of results in LS algorithm.
Recently, for stochastic systems with saturated observations, \cite{dai2025judicialsentencingpredictionbased} proposed an accelerated SG algorithm and achieve a faster convergence rate for the adaptive prediction error without relying on any PE conditions. 

In summary, most existing studies on the adaptive prediction and control for stochastic systems often rely either on strong convexity of the loss function or on a certain PE condition on the information matrix, which cannot be satisfied in many real-world systems, including feedback control systems. To the best of our knowledge, there still a lack of theoretical guarantees that characterizes the adaptive prediction and control capabilities of SG-based algorithms for nonlinear systems without requiring such strong convexity or PE data conditions.

To this end, we investigate the adaptive prediction and control performance for a class of nonlinear dynamical system under a weak convexity assumption, which can cover a broad class of nonlinear models, including saturation nonlinearities, sigmoid, ReLU and tanh activations, and standard classification models in machine learning. The main contributions of this paper are summarized as follows:

\begin{itemize}
\item Firstly, we propose a new gradient-based adaptive identification algorithm for a class of nonlinear dynamical systems, where the step-size is designed to accelerate the convergence of parameter estimates with an additional modified term to overcome the possible unboundedness of the gradient norm.

\item Secondly, we establish theoretical guarantees of adaptive prediction for the proposed algorithm. Under a weak convexity condition of the loss function, we provide the global convergence rate of the adaptive predictors without any excitation condition on the data, which makes it applicable to feedback control systems. Furthermore, it is shown that the convergence rate is faster than the classical SG algorithm.

\item Finally, based on the convergence of the adaptive predictor, we establish the convergence rate of adaptive control for a class of nonlinear stochastic systems under a classical nonlinear minimum phase condition and with a linearly-growing nonlinearity assumption, which is an extension of the framework in \cite{xie1998adaptive}.
\end{itemize}

The remainder of this paper is organized as follows. Section \ref{problem formulation}  presents the general nonlinear stochastic model, notations and assumptions. Section \ref{main results} clarifies the proposed algorithms and main theorems. Section \ref{Experiment} demonstrates the advantages of the adaptive control algorithm with a simulation example, and the proposed predictor with a real sentencing dataset.
Finally, Section \ref{conclusion} concludes this paper with some remarks.

\section{Problem Formulation} \label{problem formulation}
In this section, we will give the notations, the nonlinear stochastic model as well as the assumptions.

\subsection{Notations}
In this paper, \(\|\cdot\| \) denotes the the Euclidean norm of a matrix or vector. For a matrix \(A\), \(\lambda_{\min }\{A\}\) and \(\lambda_{\max}\{A\} \) denote the minimum  and maximum eigenvalues of \(A\) respectively, \(A^{\tau}\) denotes the transpose of \(A\). 
For two real sequences $\{a_k, k\geq0\}, \{b_k, k \geq 0\}$ with $b_k>0$, $a_k =o(b_k)$ means that $a_k/b_k \rightarrow 0$ as $k \rightarrow \infty$,  $a_k=O(b_k)$ means that there exists a positive constant $C$ such that $|a_k| \leq C b_k$ for all $k\geq0$.
In a probability space $(\Omega, \mathcal{F},P)$, $\Omega$ is the sample space, $\mathcal{F}$ is the $\sigma$-algebra of events on $\Omega$, and $P$ is a probability measure on $(\Omega, \mathcal{F})$. Besides, \(\{\mathcal{F}_k, k\geq 0\}\) represents a sequence of non-decreasing \(\sigma\)-algebras, and \(\mathbb{E}[\cdot \mid \mathcal{F}_k]\) denotes the conditional expectation. The random sequence \(\{x_k, \mathcal{F}_k\}\) is called adapted if \(x_k\) is \(\mathcal{F}_k\)-measurable for all \(k \geq 0\). 

\subsection{Model and Assumptions} \label{SCM}
We consider the following nonlinear stochastic dynamical model:
\begin{equation}\label{sclassmodel}
y_{k+1}=g(\phi_k,\theta^*,e_{k+1}),
\end{equation}
where \(g:\mathbb{R}^d \times \mathbb{R}^d \times \mathbb{R} \to \mathbb{R}\) (with \(d \geq 1\)) is a nonlinear function, \(y_{k+1} \in \mathbb{R}\) represents the observation, \(\phi_k \in \mathbb{R}^d\)  is the regression vector, \(e_{k+1} \in \mathbb{R}\) is a random noise, and \(\theta^* \in \mathbb{R}^d\) is an unknown parameter vector to be estimated. 

In order to design an adaptive identification algorithm, we consider the following prediction-based loss function:
\begin{equation}\label{loss}
\min_{\theta} \sum\limits_{k=1}^nJ_k(\theta), \quad J_k(\theta)=\mathcal{L}\left(f(\phi_k,\theta^*),f(\phi_k,\theta)\right),
\end{equation}
where \(\mathcal{L}(\cdot, \cdot)\) quantifies the loss between the optimal predictor \(f(\phi_k,\theta^*)\) based on the true parameter \(\theta^*\) and the predictor \(f(\phi_k,\theta)\) based on the estimated parameter \(\theta\). The optimal predictor \(f(\phi_k,\theta^*)\) is derived from the model \(g(\phi_k, \theta^*, e_{k+1})\) under a suitable optimality principle. Since the true parameter \(\theta^*\) is unknown \textit{a priori}, we institute $\theta^*$ by its estimate \(\theta\) in practice, which yields the predictor \(f(\phi_k,\theta)\). Naturally, our goal is to minimize the discrepancy between the optimal predictor \(f(\phi_k,\theta^*)\) and the predictor \(f(\phi_k,\theta)\), as shown in \eqref{loss}.
The following remark provides several illustrative examples of how \(f(\phi_k,\theta^*)\) can be constructed.

\begin{remark}\label{rm1}
If the loss function is the classical mean-squares loss, i.e., $\mathcal{L}(y_{k+1},x)=(y_{k+1}-x)^2$ with $x$ being $\mathcal{F}_k$-measurable, then the optimal predictor $f(\phi_k,\theta^*)= \mathbb{E}[y_{k+1}|\mathcal{F}_k]$ is the minimizer of the conditional mean square error.
Similarly, if $\mathcal{L}(y_{k+1},x)=|y_{k+1}-x|$, then the minimizer $f(\phi_k,\theta^*)$ is the conditional median of $y_{k+1}$ given $\mathcal{F}_k$, where $\mathcal{F}_k=\sigma\{\phi_t,y_t,e_t,t\leq k\}$.
Typical choices for $\mathcal{L}(\cdot,\cdot)$ also include cross-entropy, the negative log-likelihood function and the Quartile loss, leading to a corresponding optimal predictor $f(\phi_k,\theta^*)$ under the relevant optimality criterion.
\end{remark}

Based on the online-accessed data, the aim of this paper is to design an adaptive predictor to accurately predict the future output, and to further develop an adaptive controller so that the close-loop system can achieve the optimal control performance.

For the theoretical analysis, we denote the gradient noise sequence as
\begin{equation}\label{www}
w_{k+1}=\nabla_x\mathcal{L}(y_{k+1},f(\phi_k,\theta_k))-\nabla_x\mathcal{L}(f(\phi_k,\theta^*),f(\phi_k,\theta_k)),
\end{equation}
where $\nabla_x \mathcal{L}(y,x)=\frac{\partial \mathcal{L}(y,x)}{\partial x}$. Besides, we introduce the following assumptions:
\begin{assumption}\label{asm1}
The noise $\{w_{k},\mathcal{F}_k\}$ is a martingale difference sequence
with $\sup_{k\geq0}\mathbb{E}\left[w_{k+1}^2|\mathcal{F}_k\right]\triangleq \sigma^2<\infty$, and there exists $\alpha>2$ such that $\sup_{k\geq0}\mathbb{E}\left[|w_{k+1}|^{\alpha}|\mathcal{F}_k\right]<\infty$.
\end{assumption}

\begin{assumption}\label{asm2}
For each $k\geq0$, there exist positive constants $\delta, c_1, c_2$ such that
\begin{equation}\label{fun1}
\nabla J_k(\theta)^{\tau}(\theta-\theta^*)\geq \delta J_k(\theta),
\end{equation}
and
\begin{equation}\label{fun2}
(\nabla_x\mathcal{L}(f(\phi_k,\theta^*),f(\phi_k,\theta)))^2\leq c_1J_k(\theta)+c_2.
\end{equation}
\end{assumption}
\begin{remark}\label{remark1}
We note that when $\delta = 1$, condition~\eqref{fun1} is equivalent to the standard convexity condition. 
It is clear that Assumption~\ref{asm2} is weaker than the strong convexity, the $\alpha$-exp-concave assumption in \cite{hazan2007logarithmic}, and the convexity condition on the loss $\mathcal{L}$ such as quadratic, least absolute deviation, Huber and Quantile loss for linear models in \cite{ding2024class}.
Moreover, Assumption~\ref{asm2} is similar but also weaker than the quasi-convexity in \cite{hazan2015beyond} and the weak quasi-convexity assumption in \cite{hardt2018gradient}.
\end{remark}

\begin{remark}
It is easy to verify that for linear system with mean-squares loss, Assumption \ref{asm2} holds with $\delta=2$, $c_1=4$ and $c_2=0$.
Beyond this, Assumption~\ref{asm2} also accommodates a broad class of nonlinear dynamical system.  We provide several illustrative examples below.

\medskip
\noindent\textit{Example 1: Convex squared hinge loss.}

Consider a label $y_{k+1}=\operatorname{sgn}(\phi_k^{\tau}\theta^*)$ with feature vector $\phi_k\in\mathbb{R}^d$ and the true parameter $\theta^* \in \mathbb{R^d}$.
The squared hinge-loss, which is widely used in machine learning (cf., \cite{steinwart2003sparseness}), is defined as 
$$J_k(\theta)=\left(\max\left\{0, 1-y_{k+1}\phi_k^{\tau}\theta\right\}\right)^2,$$
which is clearly a convex function with respect to $\theta$. Moreover, since $J_k(\theta^*) = 0$ by definition, Assumption~\ref{asm2} is satisfied with $\delta = 1$, $c_1 = 4$, and $c_2 = 0$.

\medskip
\noindent\textit{Example 2: Mean-square loss of saturated observations.}

Consider the saturation dynamical system widely used in engineering and social science (cf., \cite{sun2004aftertreatment, wang2022applications}):
$$y_{k+1} = g(\phi_k^\tau \theta^* + e_{k+1}),$$ 
where $$g(x) = L\,\mathbb{I}_{[x \le L]} + x\,\mathbb{I}_{[L < x < U]} + U\,\mathbb{I}_{[x \ge U]},$$
with $x,L,U \in \mathbb{R}$ with $-\infty < L < U < \infty$, and the noise $e_{k+1} \sim N(0,1)$. 
Consider the loss function as $J_k(\theta)=(f(\phi_k,\theta^*)-f(\phi_k,\theta))^2$ with the optimal predictor defined as 
$$f(\phi_k, \theta^*) = \mathbb{E}[g(\phi_k^\tau \theta^* + e_{k+1}) \mid \mathcal{F}_k]\triangleq G(\phi_k^\tau \theta^*),$$
where $G(x) = U + (L-x)F_v(L-x) - (U-x)F_v(U-x) + f_v(L-x) - f_v(U-x)$, and $F_v(\cdot)$ and $f_v(\cdot)$ denote the cumulative distribution function and the probability density function of the standard normal distribution, respectively. 
The Hessian matrix of $J_k(\theta)$ is given by
$$
\begin{aligned}
\nabla^2 J_k(\theta) =& 2 [G'(\phi_k^\tau \theta)]^2 \phi_k \phi_k^\tau \\
&- 2 [G(\phi_k^\tau \theta^*) - G(\phi_k^\tau \theta)] G''(\phi_k^\tau \theta) \phi_k \phi_k^\tau,
\end{aligned}
$$ 
where $$G'(\phi_k^\tau \theta) = F_v(U - \phi_k^\tau \theta) - F_v(L - \phi_k^\tau \theta) > 0,$$ and $$G''(\phi_k^\tau \theta) = f_v(L - \phi_k^\tau \theta) - f_v(U - \phi_k^\tau \theta).$$ 
On the one hand, there exists a constant $M_1 > 0$ such that
$\nabla^2 J_k(\theta)$ is not positive semidefinite whenever
$\phi_k^\tau \theta < -M_1$, which implies that $J_k(\theta)$ is nonconvex.
On the other hand, suppose that there exists a constant $M_2 > M_1$ such that
$|\phi_k^\tau \theta| \le M_2$ for all $k \ge 0$. Then Assumption~\ref{asm2}
is satisfied with $\delta = \underline{c}$, $c_1 = 4$, and $c_2 = 0$, where
\[
0 < \underline{c}
= \inf_{|x| \le M_2}\{F_v(U - x) - F_v(L - x)\} < 1,
\]
because the distribution function $F_v(\cdot)$ is continuous and strictly increasing.

\medskip
\noindent\textit{Example 3: Cross-entropy loss of logistic regression.}

Consider the following logistic regression model for the classification tasks in machine learning (cf., \cite{bishop2006pattern}):
$$y_{k+1}=f(\phi_k^{\tau}\theta^{*})+e_{k+1},$$
where $f(x)=\frac{1}{1+\exp(-x)}$, and $e_{k+1}\sim N(0,1)$.
The loss function is chosen as
$$
\begin{aligned}
&\mathcal{L}(f(\phi_k^{\tau}\theta^{*}),f(\phi_k^{\tau}\theta) )\\
=&-f(\phi_k^{\tau}\theta^{*})\log(f(\phi_k^{\tau}\theta))-(1-f(\phi_k^{\tau}\theta^{*}))\log(1-f(\phi_k^{\tau}\theta)).
\end{aligned}
$$
Suppose that there exists a constant $M<\infty$ such that
$|\phi_k^\tau \theta| \le M$ for all $k \ge 0$, then it is clear that Assumption \ref{asm2} is satisfied with $\delta=1$, $c_1=\frac{(1+e^M)^4}{2e^{2M}}$ and $c_2=0$.

Besides the above nonlinear examples, the following one illustrates that after a suitable nonlinear transformation, a highly nonlinear model that does not satisfy Assumption \ref{asm2} can be transferred to one that does:

\medskip
\noindent\textit{Example 4: Mean-square loss of neural network with quadratic activation.}

Consider the two-layer neural network model (cf., \cite{du2018power}):
\[
y_{k+1}=\sum\limits_{i=1}^{m_1}a_i\left(\sum\limits_{j=1}^{m_2}b_{ij}(c_{ij}^{\tau}\phi_k)^2\right)^2+w_{k+1}.\]
where $a_i,b_{ij}\in\mathbb{R}$ and
$c_{ij}\in\mathbb{R}^d$ are unknown parameters.
Define $\theta_i=\sum_{j=1}^{m_2}b_{ij}\,c_{ij}\otimes c_{ij},\
\theta^*=\sum_{i=1}^{m_1}a_i\,\theta_i\otimes\theta_i,
\
\varphi_k=\phi_k\otimes\phi_k,\
\Phi_k=\varphi_k\otimes\varphi_k.$
Then by the property of Kronecker product, we obtain 
\[
y_{k+1}=\theta^{*\tau}\Phi_k+w_{k+1}.
\]
Therefore, as for the mean–square loss $J_k(\theta)$, it is easy to verify that Assumption \ref{asm2} holds with $\delta=2$, $c_1=4$ and $c_2=0$.
\end{remark}

\section{Main Results} \label{main results}
In this section, we first present the new algorithm based on the model \eqref{sclassmodel}, and then introduce the main theorems.

\subsection{Adaptive Prediction}
Let $\theta_k$ be the estimate of $\theta$ at the $k$-th step, recursively defined by the following SG algorithm:
\begin{equation}\label{sgalg}
\begin{aligned}
&\theta_{k+1}=\theta_k-\mu_k\nabla f(\phi_k,\theta_k)\nabla_x\mathcal{L}(y_{k+1}, f(\phi_k,\theta_k)),\\
&\mu_k=\frac{\mu}{r_k^{\beta_1}\log^{\beta_2}r_k+\|\nabla f(\phi_k,\theta_k)\|^2},\\
&r_k=\beta_3+\sum\limits_{t=1}^k\|\nabla f(\phi_t,\theta_t)\|^2,
\end{aligned}
\end{equation}
where the initial value $\theta_0$ can be chosen arbitrarily, and the step-size scalar $\mu\in(0,\min\{1,2\delta c_1^{-1}\})$, The parameters \(\beta_1 \in [\tfrac{1}{2},1]\), \(\beta_2 \ge 0\), and \(\beta_3 >1\) are tunable non-negative constants that can be selected according to the desired algorithmic performance. The gain \(\mu_k\) is designed to accelerate parameter convergence while maintaining stable update directions:
the term \(r_k^{\beta_1}\log^{\beta_2} r_k\) with a smaller \(\beta_1\) may speed up the convergence of the parameter estimates,  whereas the term \(\|\nabla f(\phi_k,\theta_k)\|^2\) contributes to stability by preventing the current update from becoming unbounded since the factor \(r_k^{\beta_1}\log^{\beta_2} r_k\) alone may not sufficiently constrain its growth.  For linear systems, (\ref{sgalg}) is exactly the standard SG algorithm in \cite{guo1985strong} with $\mu_k=\frac{\mu}{r_k}$.

In the following theorem, we establish the asymptotic upper bound for the prediction-based loss (\ref{loss}) of the SG algorithm (\ref{sgalg}). 

\begin{theorem}\label{theorem1}
Under Assumptions \ref{asm1}--\ref{asm2}, if $\beta_3>1$ and 
either $\beta_1=\tfrac{1}{2},\,  \beta_2>\tfrac{1}{2}$ 
or $\beta_1 \in (\tfrac{1}{2},1]$, $\, \beta_2=0$, then there exists a finite random variable $c$ such that $$\lim\limits_{n\to\infty} \|\tilde{\theta}_n\|= c<\infty, \ \hbox{a.s.,}$$
where $\tilde\theta_k=\theta_k-\theta^*$. Moreover, if for any compact set $D$, there exists a non-decreasing positive and deterministic sequence $\{d_n\}$ such that
\begin{equation*}
\sup_{\theta\in D}\|\nabla f(\phi_n,\theta)\|^2=O(d_n), \ \hbox{a.s.,} 
\end{equation*}
then we have 
\begin{equation*}\label{eq:main}
\sum_{k=1}^n \mathcal{L}(f(\phi_k,\theta^*),f(\phi_k,\theta_k)) 
= o\left(r_n^{\beta_1} \log^{\beta_2}r_n+d_n\right), \quad \text{a.s.}
\end{equation*}
\end{theorem}
\begin{pf}
From \eqref{sclassmodel} and \eqref{sgalg}, we have
\begin{equation}\label{tilde}
\begin{aligned}
\tilde\theta_{k+1}=&\tilde\theta_k-\mu_k\nabla f(\phi_k,\theta_k)\nabla_x\mathcal{L}(f(\phi_k,\theta^*),f(\phi_k,\theta_k))\\
&-\mu_k\nabla f(\phi_k,\theta_k)w_{k+1}.
\end{aligned}
\end{equation}
By Assumption \ref{asm2}, we have 
\begin{equation}\label{convex0}
\begin{aligned}
&\nabla_x\mathcal{L}(f(\phi_k,\theta^*),f(\phi_k,\theta))\nabla f(\phi_k,\theta)^{\tau}(\theta-\theta^*)\\
\geq& \delta \mathcal{L}(f(\phi_k,\theta^*),f(\phi_k,\theta)).
\end{aligned}
\end{equation}
So by (\ref{tilde}) and (\ref{convex0}), we have
\begin{equation}\label{ine1}
\begin{aligned}
&\tilde\theta_{k+1}^{\tau}\tilde\theta_{k+1}\\
=&\tilde\theta_{k}^{\tau}\tilde\theta_{k} -2\mu_k\tilde\theta_{k}^{\tau}\nabla f(\phi_k,\theta_k)\nabla_x\mathcal{L}(f(\phi_k,\theta^*),f(\phi_k,\theta_k))\\
&+ \mu_k^2\|\nabla f(\phi_k,\theta_k)\|^2\|\nabla_x\mathcal{L}(y_{k+1},f(\phi_k,\theta_k))\|^2\\
&+2\mu_k^2\|\nabla f(\phi_k,\theta_k)\|^2\nabla_x\mathcal{L}(f(\phi_k,\theta^*),f(\phi_k,\theta_k))w_{k+1}\\
&-2\mu_k\tilde\theta_{k}^{\tau}\nabla f(\phi_k,\theta_k)w_{k+1}+\mu_k^2\|\nabla f(\phi_k,\theta_k)\|^2w_{k+1}^2\\
\leq&\tilde\theta_{k}^{\tau}\tilde\theta_{k}-2\delta\mu_k\mathcal{L}(f(\phi_k,\theta^*),f(\phi_k,\theta_k))\\
&+\mu_k^2\|\nabla f(\phi_k,\theta_k)\|^2(\nabla_x\mathcal{L}(f(\phi_k,\theta^*),f(\phi_k,\theta_k)))^2\\
&+2\mu_k^2\|\nabla f(\phi_k,\theta_k)\|^2\nabla_x\mathcal{L}(f(\phi_k,\theta^*),f(\phi_k,\theta_k))w_{k+1}\\
&-2\mu_k\tilde\theta_{k}^{\tau}\nabla f(\phi_k,\theta_k)w_{k+1}+\mu_k^2\|\nabla f(\phi_k,\theta_k)\|^2w_{k+1}^2\\
=&\tilde\theta_{k}^{\tau}\tilde\theta_{k}+c_2\mu_k^2\|\nabla f(\phi_k,\theta_k)\|^2\\
&-(2\delta-c_1\mu_k\|\nabla f(\phi_k,\theta_k)\|^2)\mu_k\mathcal{L}(f(\phi_k,\theta^*),f(\phi_k,\theta_k))\\
&+2\mu_k^2\|\nabla f(\phi_k,\theta_k)\|^2\nabla_x\mathcal{L}(f(\phi_k,\theta^*),f(\phi_k,\theta_k))w_{k+1}\\
& -2\mu_k\tilde\theta_{k}^{\tau}\nabla f(\phi_k,\theta_k)w_{k+1}+\mu_k^2\|\nabla f(\phi_k,\theta_k)\|^2w_{k+1}^2.
\end{aligned}
\end{equation}
By (\ref{sgalg}), we know that 
\begin{equation}\label{e2}
\begin{aligned}
2\delta-c_1\mu_k\|\nabla f(\phi_k,\theta_k)\|^2\geq 2\delta-c_1\mu:=\underline \mu>0,
\end{aligned}
\end{equation}
and 
\begin{equation}\label{e1}
\begin{aligned}
& \sum\limits_{k=1}^{\infty}\mu_k^2\|\nabla f(\phi_k,\theta_k)\|^2
\leq\mu^2\sum\limits_{k=1}^{\infty}\frac{r_k-r_{k-1}}{r_k^{2\beta_1} \log^{2\beta_2}r_k}\\
\leq& \mu^2\sum\limits_{k=1}^{\infty}\int_{r_{k-1}}^{r_k}\frac{1}{x^{2\beta_1}\log^{2\beta_2}x} dx<\infty,
\end{aligned}
\end{equation}
where we use the fact $\beta_1=\frac{1}{2}$, $\beta_2>\frac{1}{2}$ or $1\geq \beta_1>\frac{1}{2}$, and $\beta_3>1$.

From (\ref{ine1}) and (\ref{e2}), we can obtain
\begin{equation}
\begin{aligned}
\mathbb{E}\left[\tilde\theta_{k+1}^{\tau}\tilde\theta_{k+1}|\mathcal{F}_k\right]\leq&\tilde\theta_{k}^{\tau}\tilde\theta_{k}-\underline \mu \mu_k\mathcal{L}(f(\phi_k,\theta^*),f(\phi_k,\theta_k))\\
&+(\sigma^2+c_2)\mu_k^2\|\nabla f(\phi_k,\theta_k)\|^2.
\end{aligned}
\end{equation}
Hence, by the Robbins–Siegmund theorem (cf., Theorem 1.3.2 in \cite{guo2020tima}) and \eqref{e1}, we can obtain that 
\begin{equation}\label{r1}
\|\tilde\theta_k\|\to c <\infty, \ \hbox{a.s.}
\end{equation}
and 
\begin{equation}\label{r2}
\begin{aligned}
&\sum\limits_{k=1}^{\infty}\mu_k\mathcal{L}(f(\phi_k,\theta^*),f(\phi_k,\theta_k))\\
=&\sum\limits_{k=1}^{\infty}\frac{\mu\mathcal{L}(f(\phi_k,\theta^*),f(\phi_k,\theta_k))}{r_k^{\beta_1}\log^{\beta_2}r_k+\|\nabla f(\phi_k,\theta_k)\|^2}=O(1), \ \hbox{a.s.}
\end{aligned}
\end{equation}

By Kronecker lemma (cf., Theorem 1.2.14 in \cite{guo2020tima}), we know that 
\begin{equation}
\begin{aligned}
&\sum\limits_{k=1}^n\mathcal{L}(f(\phi_k,\theta^*),f(\phi_k,\theta_k))\\
=&o\left(r_n^{\beta_1}\log^{\beta_2}r_n+\sup_{0\leq k\leq n}\|\nabla f(\phi_k,\theta_k)\|^2\right)\\
=&o\left(r_n^{\beta_1}\log^{\beta_2}r_n+d_n\right), \ \hbox{a.s.}
\end{aligned}
\end{equation}
This completes the proof.
\end{pf}
\begin{remark}
In contrast to the PE condition commonly assumed in the literature (e.g., \cite{lai1991adaptive}), which requires the matrix $\frac{1}{n}\nabla f(\phi_k,\theta)\nabla f(\phi_k,\theta)^{\tau}$ converges to a positive definite matrix, Theorem~\ref{theorem1} provides the upper bound of the prediction-based loss between the optimal predictor and the adaptive predictor without any excitation condition on the data.
Furthermore, compared to the upper bound established for the standard SG algorithm (e.g., \cite{chen2012identification}), i.e., $o(r_n)$, we are able to provide a faster convergence rate result with a modified adaptation gain scalar $\mu_k$. Specifically, we obtain $o\left(r_n^{1/2}\log^{\beta_2} r_n\right)$, provided that $\|\nabla f(\phi_k,\theta_k)\|$ is uniformly bounded, $\beta_1 = \tfrac{1}{2}$, and $\beta_2 > \tfrac{1}{2}$.
\end{remark}

\subsection{Adaptive Control}
Consider the following discrete-time nonlinear regression model:
\begin{equation}\label{modelf}
y_{k+1}=g(\phi_k,\theta^*,e_{k+1}),
\end{equation}
where $\phi_k=[y_k,\cdots, y_{k-p+1}, u_k,\cdots, u_{k-q+1}]$. Besides, the noise sequence $w_{k+1}$ in (\ref{www}) satisfies that
\begin{equation}
w_k^2=O(\alpha_k), \ \hbox{a.s.,}
\end{equation}
where $\{\alpha_k\}$ be positive non-decreasing deterministic sequence with $\alpha_{k+1}=O(\alpha_k)$.

It can be seen that under Assumption \ref{asm1}, $\alpha_k$ can be chosen as 
\begin{equation}
\alpha_k=k^{\varepsilon}, \ \forall \varepsilon\in\left(\frac{2}{\alpha},1\right),
\end{equation}
where $\alpha$ is given in Assumption \ref{asm1}.

Our objective is to design a feedback control law based on the input-output measurement to make the system outputs track a deterministic and bounded sequence $\{y_{k}^*\}$. In order to analyze this control problem, we assume that
\begin{assumption}\label{asm3}
There exists a constant $\lambda\in(0,1)$ such that 
\begin{equation}
u_{k-1}^2=O\left(\sum\limits_{t=0}^{k}\lambda^{k-t}(y_t^2+w_t^2)\right).
\end{equation}
\end{assumption}

\begin{assumption}\label{asm4}
There exist constants $K_1$, $K_2>0$ such that for bounded $\theta$,
\begin{equation}
\|\nabla_{\theta} f(\phi,\theta)\|\leq K_1+K_2\|\phi\|, \ \ \forall \phi\in\mathbb{R}^{p+q}.
\end{equation}
\end{assumption}

\begin{remark}
We remark that Assumptions \ref{asm1}-\ref{asm4} correspond to the standard assumption in the linear case (e.g., \cite{chen2012identification}). Besides, the model (\ref{modelf}) with a minimum phase condition in Assumption \ref{asm3} and a linearly-growing nonlinearity in Assumption \ref{asm4} can contain the case investigated in \cite{xie2000adaptive}. Furthermore, the linear growth condition (20) cannot be relaxed to a sup-linear condition, as has been demonstrated in Corollary 3.3.1 of \cite{xie2000adaptive}.
\end{remark}

Let $\mathcal{L}(\cdot,\cdot)$ in (\ref{loss}) denote the classical mean-squares loss used in adaptive control. From Remark (\ref{rm1}), the optimal predictor is given by $f(\phi_k,\theta^*)=\mathbb{E}\left[y_{k+1}|\mathcal{F}_k\right]$. 
Based on the parameter estimates generated by the SG algorithm (\ref{sgalg}) and according to the ``certainty equivalence'' principle, we can choose the adaptive control $u_k$ to be the solution of the equation
\begin{equation}\label{con}
f(\phi_k,\theta_k)=y_{k+1}^*.
\end{equation}

The following theorem gives the convergence rate of the closed-loop control error:
\begin{theorem}\label{theorem2}
Let Assumptions \ref{asm1}--\ref{asm4} be satisfied, and let \(\beta_1\), \(\beta_2\) and \(\beta_3\) satisfy the same conditions as those in Theorem \ref{theorem1}.  Then we have 
\begin{equation}\label{key}
\frac{1}{n}\sum\limits_{k=0}^n\left(\mathbb{E}\left[y_{k+1}|\mathcal{F}_k\right]-y_{k+1}^*\right)^2=o\left(\frac{\log^{\beta_2}n}{n^{1-\beta_1}}+\frac{\alpha_n}{n}\right), \ \hbox{a.s.}
\end{equation}
\end{theorem}
\begin{pf}
If $\sup_{n\geq0}r_n<\infty$, from Theorem \ref{theorem1}, we have 
\begin{equation}
\sum\limits_{k=1}^{\infty} \left(f(\phi_k,\theta^*)-f(\phi_k,\theta_k)\right)^2<\infty,\ \hbox{a.s.,}
\end{equation}
then we can obtain (\ref{key}).

We now consider the case $r_n\to\infty$. 
We first prove that
\begin{equation}\label{boudf}
\|\nabla f(\phi_k,\theta_k)\|^2=O(\alpha_k)+o(r_k^{\beta_1}\log^{\beta_2}r_k), \ \hbox{a.s.}
\end{equation}
For this, let 
\begin{equation}
m_k=\frac{(f(\phi_k,\theta^*)-f(\phi_k,\theta_k))^2}{r_k^{\beta_1}\log^{\beta_2}r_k+\|\nabla f(\phi_k,\theta_k)\|^2}.
\end{equation}
By (\ref{modelf}) and (\ref{con}), we know that 
\begin{equation}\label{yyk}
y_{k+1}=f(\phi_k,\theta^*)-f(\phi_k,\theta_k)+y_{k+1}^*+w_{k+1}.
\end{equation}
By Assumptions \ref{asm3}-\ref{asm4} and the boundedness of $\theta_k$, we have 
\begin{equation}\label{27}
\begin{aligned}
&\|\nabla f(\phi_k,\theta_k)\|^2\leq K_1+K_2\|\phi_k\|^2\\
=&O\left(\sum\limits_{t=0}^{k+1}\lambda^{k+1-t}y_t^2\right)+O(\alpha_k).
\end{aligned}
\end{equation}
Moreover, by (\ref{r2}), we have $m_k\to0$ almost surely, then by Assumptions \ref{asm3}-\ref{asm4}, it follows that
\begin{equation}
\begin{aligned}
&y_{k+1}^2\\
\leq& 2(f(\phi_k,\theta^*)-f(\phi_k,\theta_k))^2+O(\alpha_{k+1})\\
\leq&2m_k\left(r_k^{\beta_1}\log^{\beta_2}r_k+\|\nabla f(\phi_k,\theta_k)\|^2\right)+O(\alpha_{k+1})\\
\leq&O\left(m_k\sum\limits_{t=0}^{k}\lambda^{k-t}y_t^2+\alpha_{k+1}\right)+o\left(r_k^{\beta_1}\log^{\beta_2}r_k\right),\ \hbox{a.s.}
\end{aligned}
\end{equation}
Let $L_k=\sum\limits_{t=0}^k\lambda^{k-t}y_t^2$, then we know that there exists some constant $c>0$ such that 
\begin{equation}\label{29}
\begin{aligned}
&L_{k+1}=\lambda L_{k}+y_{k+1}^2\\
\leq& (\lambda+cm_k)L_k+O(\alpha_{k+1})+o\left(r_k^{\beta_1}\log^{\beta_2}r_k\right)\\
=&O(\alpha_{k+1})+o\left(r_k^{\beta_1}\log^{\beta_2}r_k\right), \ \hbox{a.s.}
\end{aligned}
\end{equation}
So by (\ref{27}) and (\ref{29}), we can obtain (\ref{boudf}).

Besides, by (\ref{r2}), $r_n\to\infty$ and $\|\nabla f(\phi_n,\theta_n)\|^2\leq r_n$, we know that
\begin{equation}
\Delta_n=\frac{1}{r_n}\sum\limits_{k=1}^n\left(f(\phi_k,\theta^*)-f(\phi_k,\theta_k)\right)^2=o\left(1\right), \ \hbox{a.s.}
\end{equation}

Furthermore, by Assumption \ref{asm3}, we know that
\begin{equation}
\begin{aligned}
&\frac{1}{n}\sum\limits_{k=1}^n u_{k-1}^2=O\left(\frac{1}{n}\sum\limits_{k=1}^n\sum\limits_{t=0}^{k}\lambda^{k-t}(y_t^2+w_t^2)\right)\\
=&O\left(\frac{1}{n}\sum\limits_{k=0}^n\sum\limits_{i=k}^{n}\lambda^{n-i}(y_k^2+w_k^2)\right)\\
=&O\left( \frac{1}{n}\sum\limits_{k=0}^{n}(y_{k}^2+w_{k}^2)\right)\\
=&O(1)+O\left( \frac{1}{n}\sum\limits_{k=0}^{n}y_{k}^2\right), \ \hbox{a.s.,}
\end{aligned}
\end{equation}
and by (\ref{yyk}), we have \begin{equation}\label{kk1}
\begin{aligned}
&\frac{1}{n}\sum\limits_{k=1}^n y_{k+1}^2\\
=&O(1)+O\left(\frac{1}{n}\sum\limits_{k=1}^n\left(f(\phi_k,\theta^*)-f(\phi_k,\theta_k)\right)^2\right), \ \hbox{a.s.}
\end{aligned}
\end{equation}

Besides, by Assumption \ref{asm4}, we have
\begin{equation}\label{kk2}
\begin{aligned}
&\frac{1}{n}\sum\limits_{k=1}^n\left(f(\phi_k,\theta^*)-f(\phi_k,\theta_k)\right)^2\\
\leq& \frac{\Delta_n}{n} O\left( p \sum_{k=1}^n y_k^2 
+ q \sum_{k=1}^n u_k^2 \right)\\
\leq& \Delta_n O \left( 1 + \frac{1}{n} \sum_{k=1}^n y_{k+1}^2 \right), \ \hbox{a.s.}
\end{aligned}
\end{equation}
From (\ref{kk1}) and (\ref{kk2}), it follows that
\begin{equation}
\frac{1}{n}\sum\limits_{k=1}^n y_{k+1}^2=O(1), \ \hbox{a.s.}
\end{equation}
So we have 
\begin{equation}
r_n=O(n), \ \hbox{a.s.}
\end{equation}
Then by (\ref{r2}) and (\ref{boudf}), we have 
\begin{equation}
\sum\limits_{k=1}^n\left(f(\phi_k,\theta^*)-f(\phi_k,\theta_k)\right)^2=o\left(n^{\beta_1}\log^{\beta_2}n+\alpha_n\right), \ \hbox{a.s.,}
\end{equation}
which completes the proof of Theorem \ref{theorem2}.
\end{pf}

\section{Experiments} \label{Experiment}
In this section, we will demonstrate the effectiveness of the proposed adaptive prediction algorithm via real-world-based experiments, and evaluate the control performance with a numerical simulation.

\subsection{Prediction of Judicial Sentencing based on Real-Data}

The prediction capability of the proposed algorithm is assessed using a real-world judicial sentencing dataset from cases in China. This empirical evaluation not only verifies the algorithm’s usefulness in practice but also offers methodological insights for building high-accuracy sentencing prediction methods within intelligent judiciary systems, which have become a key research focus in China’s efforts to promote judicial fairness. Moreover, given that sentencing data typically exhibit non-i.i.d. characteristics, the proposed algorithm, which requires no independence data assumptions, is naturally well-suited to this application.

Therefore, a dataset consisting of 87,588  intentional-injury cases was obtained from China Judgments Online, spanning the period from 2019 to 2024. 
Using this dataset, we compared the prediction accuracy of the proposed algorithm with that of the classical stochastic gradient algorithm (e.g., \cite{guo1985strong}). The performance evaluation was carried out based on the following practical criterion:
\begin{gather}\label{sentencing_loss}
     \frac{1}{T}\sum_{k=1}^{T} \frac{|y_{k+1}-\hat{y}_{k+1}|}{y_{k+1}},
\end{gather}
\noindent where \( T \) denotes the total count of sentencing cases.

The proposed algorithm adopts the following configuration: the noise sequence $\{\epsilon_k\}$ is independent and follows the distribution $\mathcal{N}(0,25)$, the hyperparameters are set to be $\beta_1 = 0.5$, $\beta_2 = 0.51$, $\beta_3 = 2$, and $\mu = 0.3$, and initial parameter $\theta_0$ is initialized at the zero vector. 

A comparison of the prediction accuracy (\ref{sentencing_loss}) between our proposed SG algorithm (\ref{sgalg}) and the classical SG algorithm is shown in Fig. \ref{predic error}, which clearly illustrates the superiority of our algorithm.

\begin{figure}[ht]
    \centering    \includegraphics[width=0.435\textwidth]{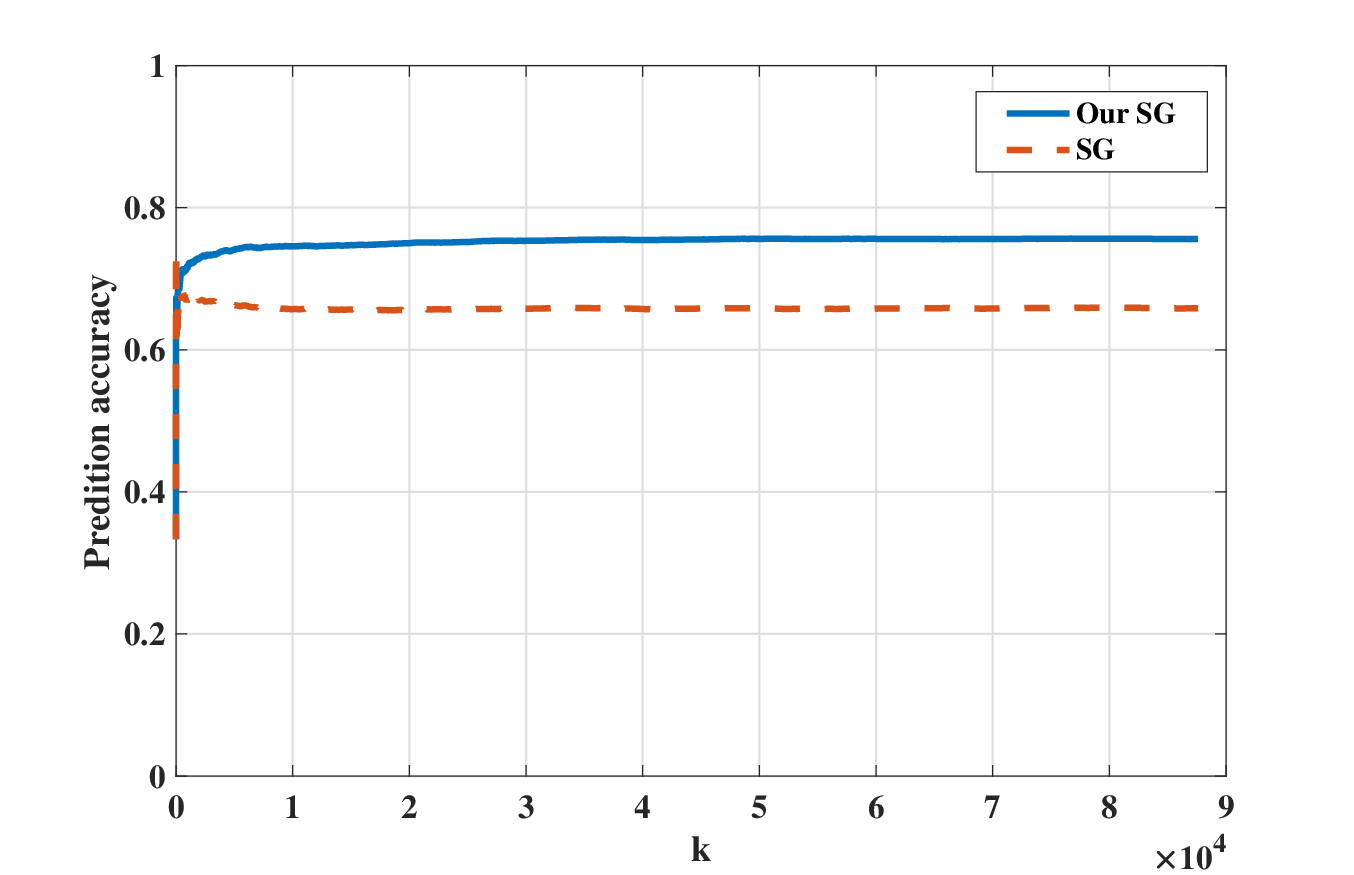}  
    \caption{Comparison of prediction accuracy.}
    \label{predic error}
\end{figure}

\subsection{Prediction and Control for a Simulation Example}
Consider the following stochastic nonlinear system:
\begin{equation}\label{control simulation}
    y_{k+1} = \tanh(\phi_k^{\tau}\theta^*) + w_{k+1},
\end{equation}
where $\tanh(x)=\frac{e^{x}-e^{-x}}{e^{x}+e^{-x}}$, $x \in \mathbb{R}$, $\phi_k=[y_k,y_{k-1}, $ $y_{k-2}$ $,u_k,u_{k-1}]^{\tau}$, and $\theta^*=[0.01,3,-0.1,0.6,-0.3]^{\tau}$. The initial conditions are $y_k=u_k=w_{k+1}=0$ for all $k<0$. The disturbance $\{w_{k+1}\}$ is i.i.d. with $w_{k+1}\!\sim\!\mathcal{N}(0,0.05^2)$. It can be verified that system \eqref{control simulation} does satisfy Assumptions \ref{asm1}–\ref{asm4}. Therefore, the data generated from this system can be used to evaluate the control performance of the proposed algorithm.

Based on the system \eqref{control simulation}, 5000 samples are generated to illustrate the evolution of the average regret and the control performance of the proposed algorithm. 
To demonstrate its advantages, we compare the prediction and control performance of our proposed SG algorithm (\ref{sgalg}) with that of the classical SG algorithm (e.g., \cite{guo1985strong}) with the step-size $\mu_k=\mu/r_k$. In the simulation, the proposed algorithm is implemented with the hyperparameters $\beta_1 = 1/2$, $\beta_2 = 2/3$, $\beta_3 = 2$, and $\mu = 0.3$. The initial parameter is set as $\theta_0 = 0.01 \times\mathbf{1}_5$, where $\mathbf{1}_5$ denotes a five-dimensional vector of ones. The control target $y_{k+1}^{*}$ is fixed at $0.5$. The SG algorithm uses the same initialization and the same step size $\mu$, with $r_0 = 1$.

As shown in Figure \ref{Average Regret}, the averaged adaptive prediction error, i.e., average regret, denoted by 
$$\frac{1}{k}\sum\limits_{i=0}^k\left(f(\phi_i,\theta^*)-f(\phi_i,\theta_i)\right)^2$$
of the predictors generated by both algorithms decreases to zero as the sample size $k$ increases. Moreover, the predictor associated with the proposed SG algorithm exhibits a faster decay of the average regret, indicating its superior performance. Figure \ref{control error} further demonstrates that the trajectory $y_{k+1}$, which depends on the estimated parameters, approaches the target $y_{k+1}^{*}$ more rapidly under the proposed SG algorithm.

\begin{figure}[ht]
    \centering    \includegraphics[width=0.36\textwidth]{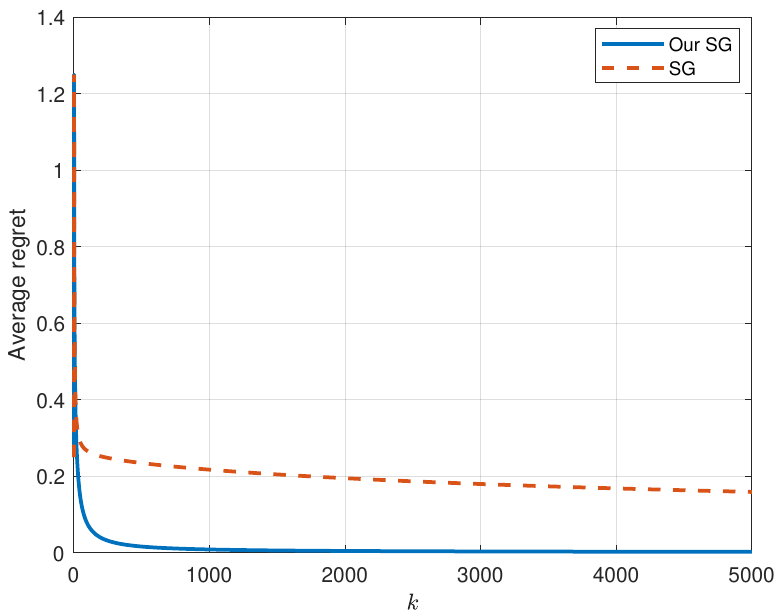}  
    \caption{Comparison of average regret.}
    \label{Average Regret}
\end{figure}

\begin{figure}[ht]
    \centering    \includegraphics[width=0.36\textwidth]{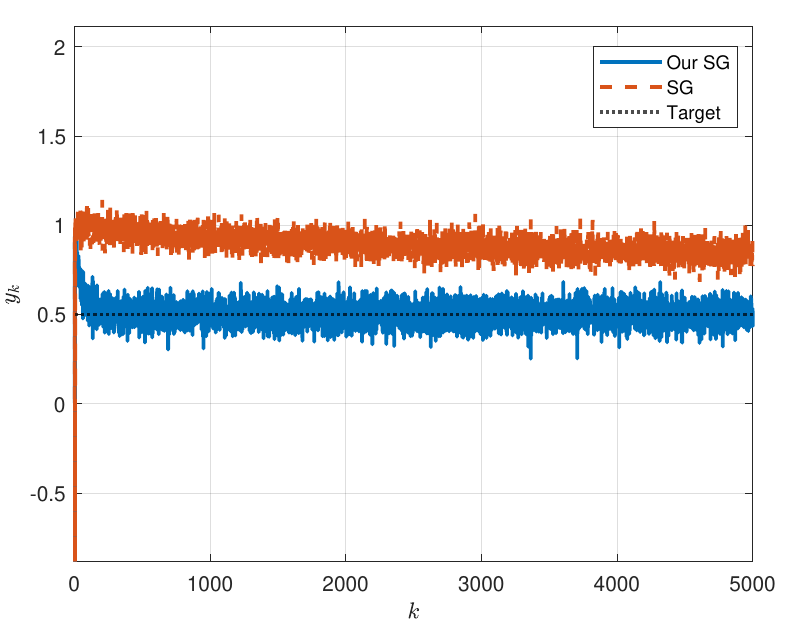}  
    \caption{Comparison of control performance.}
    \label{control error}
\end{figure}

\section{Conclusion} \label{conclusion}
In this paper, we have investigated the performance of adaptive prediction and control for a class of nonlinear stochastic systems under a weak convexity condition. 
To accelerate the convergence rate of the classical stochastic gradient algorithms, we propose a novel nonlinear adaptive identification algorithm to estimate the unknown true parameter. Without any excitation condition on the regressor data, we have established the convergence rate of both the averaged prediction error and the adaptive control error, which is much faster than the classical stochastic gradient algorithm for linear stochastic systems. These findings are consistently supported by numerical simulations as well as real-data experiments. For further investigation, it would be interesting to explore the weakest possible excitation condition for the consistency of our SG algorithm with $\beta_1\in(\frac{1}{2},1)$.

\bibliography{ifacconf}        
\end{document}